\documentclass[11pt]{article}
\usepackage{amsmath, amssymb, amscd, amsthm, amsfonts}
\usepackage{graphicx}
\usepackage{hyperref}
\usepackage{xcolor}
\usepackage{cleveref}
\usepackage{parskip}
\usepackage{mathabx}
\usepackage{tikz}

\usepackage[T1, T2A]{fontenc}
\usepackage[english, russian]{babel}

\usepackage[title]{appendix}

\title{\textbf{A new look at the theory of  point interactions}}
\author{Rodolfo Figari\
\\
Gran Sasso Science Institute, {\it Viale Francesco Crispi, 7 -  67100 LAquila (AQ)}  \\
rodolfo.figari@protonmail.com 
\\
\\
Hamidreza Saberbaghi
\\
Gran Sasso Science Institute, {\it Viale Francesco Crispi, 7 -  67100 LAquila (AQ)} \\
hamidreza.saberbaghi@gssi.it
\\
\\
Alessandro Teta\
\\
Sapienza Universit\`a di Roma, {\it Piazzale Aldo Moro, 5 - 00185 (RM), Italy}\\
teta@mat.uniroma1.it\\}
\date{}

\newtheorem{theorem}{Theorem}[section]
\newtheorem{proposition}{Proposition}[section]

\newtheorem{lemma}[theorem]{Lemma}

\newtheorem{remark}{Remark}
\newcommand{\be}{\begin{equation}}
\newcommand{\ee}{\end{equation}}
\newcommand{\rr}{\mathbb{R}}
\newcommand{\cc}{\mathbb{C}}

\newcommand{\real}{\operatorname{Re}}
\newcommand{\imag}{\operatorname{Im}}
\newcommand{\n}{\noindent}
\newcommand{\vs}{\vspace{0.5cm}}

\crefname{appendix}{Appendix}{Appendices}

\usepackage{mathrsfs}


\begin{document}
\selectlanguage{english}
\maketitle

\begin{abstract} 
We investigate the entire family of multi-center point interaction Hamiltonians. We show that a large sub-family of these operators does not become either singular or trivial when the positions of two or more scattering centers tend to coincide. In this sense, they appear to be renormalized by default as opposed to the "local" point interaction Hamiltonians usually considered in the literature as the ones of physical interest. In the two-center case, we study the behaviour of the negative eigenvalues as a function of the center distance. The result is used to analyze a formal Born-Oppenheimer approximation of a three-particle system with two heavy and one light particle. We show that this simplified model does not show any ultraviolet catastrophe, and we prove that the ratio of successive low energy eigenvalues  follows the Efimov geometrical law.

\end{abstract}

\section{Introduction}
Multi-center zero-range potentials have been extensively used in atomic physics as models of interaction between charged or neutral particles and arrays of atoms (see e.g. \cite{dem}). Their formal definition has been given through singular boundary conditions around each point scatterer, generalizing the single-point  delta interaction (also referred to as zero-range, or contact interaction). A rigorous definition of quantum Hamiltonians with singular potentials supported on thin sets (e.g. discrete sets of points or manifolds of dimensions strictly lower than the configuration space) was only given in the second half of the last century. The definition was based on the theory of self-adjoint extensions of symmetric operators, of which various alternative versions have since been given. 
Despite the generality of the theory, physical intuition led to the investigation of subsets of extensions connected to particular boundary conditions. This paper deals with some pathological aspects of these last mentioned multi-center point interaction Hamiltonians.
In fact, such pathologies seem to reverberate 
in ultraviolet singularities emerging in the theory of many quantum particles interacting via zero-range forces. \\
The standard definition of multi-center point interaction Hamiltonians is based on the assumption that functions in the Hamiltonian domain should satisfy (singular) boundary conditions on each interaction point, which do not depend on the position of the other ones. For this reason they are referred to as local point interaction Hamiltonians. It was recognized (see e.g. \cite{Albeverio}, \cite{Fassari}, \cite{Mosta}) that the assumption resulted in a non-additive property of point interactions. In particular, a local multi-center point interaction with fixed strength parameters acts as  sum of single point interactions only when points are well spaced, but produces  ground states of lower and lower energy when two or more points get closer and closer.\footnote{As an example, in \cite{FHT} it was proved that any Schr\"{o}dinger Hamiltonian with a $L^1 (\rr^3)$ potential was the limit, in the resolvent sense, of a sequence of Hamiltonians with an increasing number of rescaled local point interactions. The limit potential turned out to be the asymptotic scattering length density of the point scatterers. However the result could be proved only for distribution of points avoiding close contact among them}\\
As it is well known, the same ultraviolet problem arises in the theory of three or more non-relativistic bosons interacting via contact interactions (\cite{TH},\cite{mf1},\cite{mf2}).  Moreover, a formal application of the Born-Oppenheimer approximation to the three-particle case, which makes use of local point interactions between the light particle and the heavy ones, shows the same ultraviolet catastrophe.\\
In the following sections, we recall the characterization given in \cite{DG} of all multi-center point interactions and we show that almost all of them have additive properties for any relative positions of the point scatterers. Read in terms of the theory of local interactions, the regularization appears as a  renormalization of the boundary condition of the kind suggested in \cite{mf1},\cite{AHKW} for the three-particle case, which have been recently made rigorous by various authors (\cite{FiT},\cite{Miche},\cite{BCFT},\cite{FeT}) to solve the stability problem.\\
In the last section we analyze thoroughly the general two center point interaction Hamiltonians and then we use the results to investigate the  Born-Oppenheimer approximation of the three-particle problem.  We prove that no ultraviolet catastrophe takes place and there are interactions with infinite scattering length (unitary limit) providing effective Hamiltonians  with an infinite number of Efimov states corresponding to non-degenerate eigenvalues accumulating to the continuum threshold. We show that the number of Efimov states becomes finite and decreases as the interaction parameters move away from the one corresponding to the unitary limit. In \cite{fons}, Fonseca et al. obtained similar results using separable potentials. We investigate the same subject within the realm of zero-range interactions.

\section{Point interactions in $\mathbb{R}^3$}

One can find several books and articles that describe exhaustively the history of zero-range potentials in Quantum Mechanics and other fields of Mathematical Physics. We only mention the introduction of the book \cite{Albeverio} due to the rich list of bibliographical references  it contains. \\
We will not add another attempt in this sense. In the following,  we merely want to briefly recall the different definitions given in the past of these interactions.

\subsection{The one center point interaction in $\rr^3$}
Below we list some of the ways used in the literature to define a point interaction Hamiltonian in $\rr^3$.
\begin{itemize}
\item In time independent scattering theory for a central potential it is proved that the scattering cross section for S-wave scattering is given by $\tfrac{4 \pi}{k^2} \sin^2 \delta_{0}(k)$, 
where $\delta_{0}$ is the S-wave phase shift in the expansion of the scattering amplitude. A non-trivial scattering cross section in the small energy limit is then characterized by the fact that the scattering length 
\[
\displaystyle{a} :=  -\lim_{k \to 0} \frac{\sin \delta_{0}(k)} {k}
\]
 results different from zero.\\
On the other end, the generalized eigenfunctions of any Schr\"odinger Hamiltonian with central short-range potential should behave in the region where the potential is zero as 
\[
\Psi (r) = q \frac{\sin (kr + \delta_{0})}{r}\]
where $q \in \cc$ . If the region where the potential is zero extends to $\rr^3\setminus\{0\}$ one immediately checks that 
\begin{equation} \label{BC}
\lim_{k\rightarrow 0} k \cot \delta_0(k) =-\frac{1}{a} = \left[\frac{d}{dr} (\log (r \Psi))\right]_{r=0} .
\end{equation}
The "boundary condition" (\ref{BC}) together with the definition of scattering length uniquely characterize   the behaviour of the wave functions (in particular of eigenfunctions  and generalized eigenfunctions) close to the point scatterer up to the zeroth order in $r$
\[ \psi  \sim \frac{q}{r} -  \frac{q}{a} \quad \,\,\,\,\,\,\,\,\, \mbox{for} \,\,\,\,\, r \rightarrow 0\,.\]
Boundary condition (\ref{BC}) was taken  as definition of a single point interaction of scattering length $a$ placed in the origin. For the relation of this definition with the "Fermi pseudopotential" (see, e.g., \cite{Albeverio}, notes to the chapter one.)

\item
Within the general theory of perturbations of the Laplacian  supported on "small" sets, the Hamiltonian with a zero-range potential in the origin is defined as any self-adjoint extension of the symmetric operator $\dot{H}_0$ one obtains restricting the free Hamiltonian $ - \Delta $ to  $C^{\infty}_{0}(\rr^{3}\setminus \{0\})$. In this simple case the defect spaces $N_{z}$ of $\dot{H}_0$ for $z \in \cc\setminus \rr^+$  are one dimensional. In fact, the only eigenfunction in $L^2(\rr^3)$ relative to the eigenvalue $z$ of the adjoint operator $\dot{H}^{\star}_0$ is the function
\begin{equation} \label{Green's}
G_z =  \frac{e^{-\sqrt{-z}|x|}}{ 4 \pi |x|} \,\,\,\,\,\,\,  \,\,\,\,\,\, \real \sqrt {-z } \geq 0 \,.
\end{equation}
According to the von Neumann construction, each self-adjoint extension can then be parametrized by the phase $\phi \in [0 , 2\pi)$  associated with the unitary operator $U_\phi$ from $N_{i}$ to $N_{-i}$ ($U_\phi \psi = e^{i \phi} \psi \,,\,\,\,\ \forall \psi \in N_{i}$).

\item
Dirichlet forms and Non-Standard Analysis provided alternative definitions of point interaction Hamiltonians which were particularly important in the theory of stochastic processes and stochastic fields associated to quantum mechanical systems.  We refer to the appendices F and H of \cite{Albeverio} for details.
 
\end{itemize}

In the case of one center    $y \in \rr^3$ all the strategies mentioned above brings to alternative characterisations of the same family of self-adjoint Hamiltonians whose properties are listed here below. \\
The family $H_{\alpha,y}$ is indexed by a real parameter $\alpha$ connected with the scattering length $a$ and with the unitary operator $e^{i\phi}$ in the von Neumann construction  via the formula
\begin{equation*}
    \alpha = - \frac{1}{4 \pi a} = \frac{1}{4 \pi} \frac{1}{\sqrt{2}}  \left(  \tan\frac{\phi}{2} - 1 \right).
\end{equation*}
Domain and action of the Hamiltonian $H_{\alpha,y}$  are the following
\begin{equation}\label{doH}
D(H_{\alpha,y}) = \left\{u \in L^{2}(\rr^{3})| u = \phi^{\lambda} + q 
G^{\lambda}(\cdot -y),\;\; \phi^{\lambda} \in H^{2}(\rr^{3}), \,\,\, q \in \cc, \,
\;\; \phi^{\lambda}(y) = \left(\alpha + \frac{\sqrt{\lambda}}{4 \pi} \right) q
\right\} 
\end{equation}
\begin{equation}
(H_{\alpha,y}+\lambda) u = (-\Delta + \lambda) \phi^{\lambda}
\end{equation}
for $\lambda >0$, $\lambda \neq 16 \pi^{2} \alpha^{2}$ when $\alpha < 0$ and 
\[ 
\displaystyle G^{\lambda} (x) = \frac{e^{-\sqrt{\lambda}|x|}}{ 4 \pi |x|}  
\] 
is the Green's function of $(-\Delta + \lambda), $ where $ G^{\lambda}  =  G_{z =-\lambda}$. 
In \cite{Albeverio}  the following properties of $H_{\alpha,y}$ are proved.\\
- The resolvent $(H_{\alpha,y} -z)^{-1}$  for $z =-\lambda$, with $\lambda >0$, $\lambda \neq 16 \pi^{2} \alpha^{2}$ when $\alpha < 0$, reads

\begin{equation}\label{re1}
[(H_{\alpha,y}+\lambda)^{-1} f](x) = (G^{\lambda}f)(x) +  
\frac{1}{\alpha + \sqrt{\lambda}/4 \pi}G^{\lambda}(x-y) 
(G^{\lambda}f)(y)
\end{equation}
where, with an abuse of notation, we indicated with $G^{\lambda}f$  the convolution of the function $G^\lambda$ with a function $f \in L^2(\rr^3)$.  
Notice that the free Hamiltonian $-\Delta$ is obtained in the limit $\alpha \rightarrow \pm \infty$.

- the spectrum of $H_{\alpha, y}$ is
\[ \sigma (H_{\alpha, y}) = [ 0,\infty) \;\;\;\; \alpha \geq 0 \]
\[\sigma (H_{\alpha, y}) = \left\{-16 \pi^{2} \alpha^{2}\right\} \cup[ 0,\infty) 
\;\;\;\; \alpha < 0 
\]

where the continuous part $[0,\infty)$ is purely absolutely continuous.

- for $\alpha < 0$ the only eigenvalue is simple and the corresponding 
normalized eigenfunction is
\[
\psi_{\alpha}(x) = \sqrt{- 2 \alpha} \frac{\exp(4 \pi \alpha |x|)}{|x|}
\]

- for any $\alpha \in \rr$,  for each $k \in \rr^3$  the  function 

\begin{equation}
\psi_{\alpha}(k,x) = \frac{1}{(2 \pi)^{3/2}} \left(\exp(ik \cdot x) 
-\frac{1}{\alpha- i |k|/ (4 \pi)} \frac{\exp(i|k| |x|)}{|x|}\right)
\end{equation}
is a generalized eigenfunction of $H_{\alpha, y}$ relative to the energy  $ E = |k|^{2} $.

- The scattering length $a=-1/4\pi \alpha$ associated to the operator $H_{\alpha,y}$ is negative for $\alpha$ positive and viceversa. In particular the scattering length is infinite for $\alpha=0$.

- In the origin functions $u$ in  the domain satisfy a boundary
condition expressed by the last equality in (\ref{doH}). 
If we define $r=|x|$ it is easy to see that the boundary condition
can be equivalently written
\begin{equation}
\lim_{r \rightarrow 0} \left[ \frac{\partial(ru)}{\partial r} - 4 \pi
\alpha (r u) \right]=0, \;\;\;\;
\label{bc1}
\end{equation}
which coincides with \eqref{BC}.\\ 
In summary, the one-center point interaction is a simple and versatile model of a single particle quantum dynamics. In particular for $\alpha = 0$ the interaction has zero effective range and infinite scattering length. We will come back to this important feature of point interaction Hamiltonians.

\subsection{The ``local" many-center point interaction Hamiltonians in $\rr^3$}
The quantum dynamics of one quantum particle in $\rr^3$ in presence of $n$ point scatterers  was first analysed in atomic physics using formal definitions of the Hamiltonians given in terms of boundary conditions generalizing (\ref{BC}).  In chapter II.1 of  \cite{Albeverio} those definitions were made precise and their properties rigorously proved. We give here a brief summary of the subject. \\
For any $\underline{\alpha} = \{\alpha_{1}, \ldots , \alpha_{n}\}$, 
with $\alpha_{i} \in \rr, \;\;i = 1,\ldots,n$, and $\underline{y} = \{y_{1}, 
\ldots , y_{n}\}, \;\;y_{i}\in \rr^{3}, \;\;i =1,\ldots,n$, the operator 
$H_{\underline{\alpha}, \underline{y}}$ defined by 
\begin{equation}\label{domop3}
    \begin{split}
        D(H_{\underline{\alpha},\underline{y}})= &\left\{u \in L^{2}(\rr^{3}) 
\,|\, u =\phi^{\lambda} + \sum_{k=1}^{n} q_{k} 
G^{\lambda}(\cdot-y_{k}),\;  \phi^{\lambda} \in H^{2}(\rr^{3}), \right. \\
    & \left. \,\,\,\, q_k \in \cc,\; \phi^{\lambda}(y_j) = \sum_{k=1}^{n} [\Gamma_{\underline{\alpha},\underline{y}}(\lambda)]_{jk} \;q_{k}, 
\; j=1,...,n \right\}
    \end{split}
\end{equation}
\begin{equation}
(H_{\underline{\alpha},\underline{y}} + \lambda)u= (-\Delta + \lambda)
\phi^{\lambda}
\label{aziop3}
\end{equation}
for  $q_i \in \cc$, $\lambda >0$  and   
\begin{equation}
[\Gamma_{\underline{\alpha},\underline{y}}(\lambda)]_{jk}= \left( \alpha_j  +
\frac{\sqrt{\lambda}}{4 \pi} \right) \delta_{jk} - G^{\lambda}(y_j - y_k)(1 -
\delta_{jk}),
\label{gamma3}
\end{equation}
is a self-adjoint extension of $\dot{H}_0 = -\Delta \restriction  C^{\infty}_{0}(\rr^{3} \setminus \{y_{1},\ldots,y_{n}\})$. Such extension is referred to as the "local"  
point interaction Hamiltonian with $n$ centers of strength $\alpha_i$ located in the points $y_i$.
Notice that for any smooth function $u \in D(H_{\underline{\alpha},\underline{y}})$ 
vanishing at each $y_i$, one has $q_i =0\,\, \forall i$  and then, from (\ref{aziop3}), $H_{\alpha,y}u= -\Delta u$.\\ 
At each point $y_j$ function $u$ in  the domain satisfy a boundary
condition expressed by the last equality in (\ref{domop3}). 
If we define $r_j=|x-y_j|$ it is easy to see that the boundary condition
can be equivalently written as 
\begin{equation}
\lim_{r_j \rightarrow 0} \left[ \frac{\partial(r_j u)}{\partial r_j} - 4 \pi
\alpha_j (r_j u) \right]=0, \;\;\;\;j=1,...,n\,
\label{bc}
\end{equation}
which is a direct generalization of \eqref{bc1} to the  multi-center case. \\
For $z =-\lambda$, with $\lambda >0$ sufficiently large so that $\det \Gamma_{\underline{\alpha},\underline{y}}(\lambda)\neq 0$, the resolvent $(H_{\underline{\alpha},\underline{y}} + \lambda)^{-1}$ reads
\begin{eqnarray}
&&[(H_{\underline{\alpha},\underline{y}} + \lambda)^{-1}f](x) = (G^{\lambda} f)(x)+
\sum_{j,k=1}^{n} [\Gamma_{\underline{\alpha},\underline{y}}(\lambda)^{-1}]_{jk}
G^{\lambda}(x-y_j) (G^{\lambda}f)(y_k).
\label{resol3}
\end{eqnarray}
for $f \in L^2(\rr^3)$.
A thorough investigation  of (\ref{resol3}) gives the following characterization of the spectral properties. \\
- The continuous spectrum of $H_{\underline{\alpha},\underline{y}}$ is purely absolutely continuous and coincides with the positive real axis, \\ 
- the point spectrum consists of (at most) $n$ non-positive eigenvalues given by the possible solutions $E \leq 0$ of the
equation $\det \; [\Gamma_{\underline{\alpha},\underline{y}}(-E)]=0$. \\

Notice that the set of Hamiltonians $H_{\underline{\alpha},\underline{y}}$ 
 is a strict subset of the set of all the self-adjoint extensions of $ \dot{H}_0 = -\Delta \restriction C^{\infty}_{0}(\rr^{3} \setminus \{y_{1},\ldots,y_{n}\})$. Indeed, the dimension of the defect spaces $N_{z}$ for $\imag(z) \neq 0$ in this case is $n$. As a consequence, $n^{2}$ real parameters are required to identify all unitary operators from  $N_{z}$ to $N_{\bar{z}}$. The ``missing ones" were denoted ``non-local'', meaning that the functions belonging to the corresponding Hamiltonian domains obey boundary conditions around each point depending on the positions  of the other points.\\
Few remarks are worth at this point. To simplify notation we will examine only cases in which all the strength parameters are equal ($\alpha_{i} = \alpha \,\, \forall i$) and all the functions in the Hamiltonian domain are completely symmetric under exchange of the scatterer positions ($q_{i} =q \,\,\forall i$). We will come back to this symmetry requirement later. \\
- The boundary condition in the second line of (\ref{domop3}) is strongly dependent on the geometry of the point scatterer configuration. In particular, for two point scatterers at distance $r$ 
there is only one eigenvalue approaching $-\infty$ when $r$
goes to zero. In fact, in this case the matrix (\ref{gamma3}) reduces to 
\begin{equation}\label{eigeneq}
    \Gamma_{\alpha, r}(\lambda) = \left(\begin{array}{cc}
 \left( \alpha +\frac{\sqrt{\lambda}}{4 \pi} \right) & - G^{\lambda}(r) \\
 - G^{\lambda}(r) & \left( \alpha +\frac{\sqrt{\lambda}}{4 \pi} \right)
\end{array}\right)
\end{equation}
and the equation $\det \Gamma_{\alpha,r}(\lambda) =0$ 
has only one solution 
$\lambda$ increasing to $+\infty$ ( like $1/r^2$) when $r$ tends to $0$. Moreover for any $\lambda >0$  the non-diagonal terms of 
the matrix {$\Gamma_{\alpha,r}(\lambda)$} go  to infinity  as $r$ tends to $0$. In turn,  the  Hamiltonian converges in the resolvent sense to the free Hamiltonian, i.e.,  
the two-point scatterers disappear.\\
- Another sign of such ``non-additivity pathology''  is the fact that the scattering length $ a_{\alpha,r}$ 
vanishes in the limit $r \rightarrow 0$, as it is clear from the formula  (see \cite{Albeverio})
\begin{equation*}
    a_{\alpha, r}= -\frac{1}{4\pi} \sum_{j,k=1}^2 [\Gamma_{\alpha,r} (0)^{-1} ]_{jk} \,.
\end{equation*}
- It is easy to realize that a fixed boundary  condition for two close  scatterers is untenable if we want to avoid  trivial  limit when their distance approaches zero. In fact, consider a point $x$ belonging to the segment connecting $y_1$ and $y_2$,  i.e. $\displaystyle x =  \gamma y_2 + (1-\gamma) y_1 , \,\,\,\gamma \in (0,1)$. For $r =|y_1 - y_2| \rightarrow 0$  and any $u \in D( H_{\alpha, y_1, y_2})$ one would have
\begin{equation*}
    u(x) = \frac{q}{|y_{1} - x|}-\frac{q}{a} + O(1)  = \frac{q}{\gamma r}-\frac{q}{a} + O(1)
\end{equation*}
and at the same time
\begin{equation*}
    u(x) = \frac{q}{|y_{2} - x|}-\frac{q}{a} + O(1)  = \frac{q}{(1-\gamma) r}-\frac{q}{a} + O(1)
\end{equation*}
which cannot both be true for any $\gamma$ unless $q$ is equal to zero. \\
The situation is reminiscent of what happens in the more complicated case of three bosons interacting via zero-range forces. Under the hypothesis of a boundary condition of the kind written above, Thomas \cite{TH} noticed that the ground state of the system had an infinite negative energy (the ``fall to the center'' problem in the three-boson system). In \cite{mf1},\cite{AHKW}, the authors suggested that a way out of this pathology consists in  adding a term in the boundary condition that becomes infinite at the coincidence of the three-particle positions. We will see in the next section that this is true also in  the one-body problem  with a finite number of centers.

\subsection{The ``non-local'' many center point interaction Hamiltonians in $\rr^3$}
In \cite {DG} L. Dabrowskj and H. Grosse used the von Neumann construction to characterize \emph{all} the self-adjoint extensions of $\dot{H}_{0} = -\Delta \restriction  C^{\infty}_{0}(\rr^{d} \setminus \{y_{1},\ldots,y_{n}\})$ in dimension $d$ equal one, two and three.  Already in the title they used the word non-local for such extensions in spite of the fact that they were examining the set of all extensions. 
 We summarize their results in $d = 3$ changing only some notation in order to simplify the comparison with what was reviewed in the previous subsection. For further details and for the proofs of the results, we refer to the original paper.\\
The defect spaces $N_{i}$ and $N_{-i}$ are $n$ dimensional and are respectively the linear span of the solutions in $L^2 (\rr^3)$ of the fundamental equation $(-\Delta \mp i) G_{\pm i}^{j} = \delta_{y_{j}}$ which are eigenfunctions of the adjoint of $\dot{H}_0$ relative to the eigenvalues $\pm i$ . Explicitly
\begin{equation}
G_{\pm i}^{j} = \frac{e^{i \sqrt{\pm i}|x -y_{j}|}}{4\pi|x - y_{j}|} \quad \, j=1 \ldots n \qquad
\mbox {or in Fourier space} \qquad g_{\pm i}^j =  \Hat{G}_{\pm i}^{j}= \frac{1}{(2\pi)^{3/2}}
\frac{e^{i p \cdot y_{j}}}{|p|^{2} \mp i}  \nonumber
\end{equation}
with $\imag  \sqrt{\pm i} > 0$.
The self-adjoint extensions are in one to one correspondence  with  unitary operators  from  $N_{i}$ to $N_{-i}$. The extension associated with the unitary operator $V$ has the following domain and action
\begin{eqnarray} \label{von-domain-action}
     f &= &f_0 + F_{i} + V F_{i} \nonumber \\
     H_V f &= &- \Delta f_0 + i\,F_{i} - i \,V  F_{i} \nonumber
\end{eqnarray}
 where $f_{0} \in H^{2}(\rr^{3}),\,\,\, f_0 (y_j) =0 \,\, \forall j$ and $ F_{i} \in N_{i}$. Any unitary operator $V$ can be written via a matrix $U$ such that
\begin{equation} \label{unitary-action}
 V G_{i}^{l} = \sum_{j=1}^{n}U_{jl} G_{-i}^{j}
\end{equation}
where the matrix $U$ (which is non-unitary being the basis of the $G_{-i}^{j}$ non-orthonormal) has to satisfy
\begin{equation} \label{unicondition}
U^{*}S(i,i)U^{T} = S(i,i) 
\end{equation}
with $S_{mn} (\Bar{z}{,}w) = \left(\overline{G^m_{z}} {,} G^n_w \right)$.  
The resolvent of the operator $H_{V}$ is computed via  Krein's formula and reads
\begin{equation} \label{resolvent}
    \left(H_V - z \right)^{-1}=  G_z +  \sum^{n}_{k,j =1}  \left[\Gamma_V (z)^{-1}\right]_{kj} \left( \overline{G_{z}(\cdot-y_k)}{,}\cdot \right) \, G_z(\cdot-y_j)
\end{equation}
where the matrix $ \Gamma_V(z)$ must satisfy the unitarity condition:
\begin{equation} \label{Matrix-Gamma}
\Gamma_V(z) =  2i S(i,i)(U^{T}+I)^{-1} - (z+i)S(\Bar{z},-i).
\end{equation}
Eigenvalues and resonances of $H_{V}$ are determined by the condition  $\det \Gamma_V(z) = 0$.\\
An explicit computation of the matrix elements appearing in \eqref{Matrix-Gamma} gives
\begin{gather}
\begin{aligned}\label{notation}
S_{kj}(i,i) &= \frac{1}{4 \sqrt{2}\pi} \quad && \text{for } k = j \\
S_{kj}(i,i) &= \frac{1}{4 \pi r_{kj}} e^{-r_{kj} / \sqrt{2}} \sin\left(\frac{r_{kj}}{\sqrt{2}}\right) \quad && \text{for } k \neq j \\
-(z+ i)S_{kj}(\bar{z},i) &= \frac{\sqrt{-z}-\sqrt{i}}{4 \pi} \quad && \text{for } k = j \\
-(z+ i)S_{kj}(\bar{z},i) &= -\frac{e^{-\sqrt{-z} \, r_{kj}}-e^{-\sqrt{i}\, r_{kj}}}{4 \pi r_{kj}} \quad && \text{for } k \neq j \\
\end{aligned}
\end{gather}

where $r_{kj} = |y_{k} - y_{j}|$. 
Notice that all the matrix elements have regular limits for $r_{kj} \rightarrow 0$. For fixed $U$ two-point interaction scatterers do not ``disappear" when their distance decreases. This. in turn, says that the correspondence between multi-center local point interactions and unitary operators in the von Neumann construction is non-local and becomes singular when the positions of the two point scatterers coincide.\\
Pursuing the main purpose of this work, we will analyze only point interaction Hamiltonians with two-point scatterers in the symmetric case, i.e. when the functions in the domain are symmetric under the exchange of interaction point positions. We plan to return to the general multi-center case in further work.

\section{The two-center point interaction in the general case}

\subsection{Characterization of the Hamiltonian}

 The symmetry assumption implies that the only vectors in $N_{\pm i}$ belonging to the domain of $H_V$ are complex multiple of the sum $( G_{\pm i}^{1} + G_{\pm i}^{2})$. Together with the unitarity condition this entails that the $2\times2$ matrix $U$ defining the unitary operator $V$ must have the form
\begin{equation} \label{unitary-twobytwo}
U = \left(\begin{array}{ll}
e^{i \theta} & 0 \\
0 & e^{i \theta}
\end{array}\right) 
\qquad  \text{where}
\qquad 0 \leq \theta < 2 \pi\,.
\end{equation}
We will collect the main properties of the self-adjoint extension $H_\theta$ explicitly calculated from the general formulas in \cite{DG}
in the following lemma.
\begin{lemma}
To  the  matrix $U$ of the form  \eqref{unitary-twobytwo}  is  associated  the  self-adjoint extension $H_\theta$ of  $\dot{H}_{0} = -\Delta \restriction  C^{\infty}_{0}(\rr^{3} \setminus \{y_{1},y_{2}\})$ with the following properties
\begin{equation} \label{resolvent2}
    \left(H_\theta - z \right)^{-1}f=  G_z f+  \sum^{2}_{m,n =1}  \left[\Gamma_\theta(z)\right]^{-1}_{mn}  \,\overline{G_{z}
  f(y_m) }  \, G_z(\cdot -y_n)
\end{equation}

where  \eqref{Matrix-Gamma} is easily made explicit 

\begin{equation} \label{Gamma-expression}
  \Gamma_\theta(z)=
\left(\begin{array}{ll}
( i+t_\theta)s + c_{z} & ( i+t_\theta)S + C_{z} \\
( i+t_\theta)S + C_{z} & ( i+t_\theta)s +  c_{z}
\end{array}\right)
\end{equation}
with 
\begin{equation}
s =\frac{1}{4 \sqrt{2}\pi } \qquad 
S =\frac{1}{4 \pi r} e^{-r / \sqrt{2}} \sin{(r/\sqrt{2})} \qquad
c_{z} =\frac{\sqrt{-z}-\sqrt{i}}{4 \pi} \qquad 
C_{z} =-\frac{e^{-\sqrt{-z}\, r}-e^{-\sqrt{i}\, r}}{4 \pi r} 
\end{equation}
and 
\begin{equation}
t_\theta= \tan\left(\frac{ \theta}{2}\right)\,, \;\;\;\;\;\; r =  |y_1-y_2|\,.
\end{equation}
For fixed $t_{\theta}$, the resolvent \eqref{resolvent2} converges in norm  when $r \to 0$ to the resolvent of a one-center point interaction with integral kernel
\begin{equation} \label{two-becomes-one}
 G_z f + \frac{4\pi}{\sqrt{-z} + \frac{t_{\theta}-1}{\sqrt{2}}} \overline{G_{z}f(y)}
    G_{z}(\cdot-y)
\end{equation}
The eigenvalue-resonance  equations $\det \Gamma_\theta (z) = 0$ read
\begin{align} \label{twoeigen1}
    \sqrt{-z}\,r - \frac{r}{\sqrt{2}} + t_{\theta} (\frac{r}{\sqrt{2}} + e^{-\frac{r}{\sqrt{2}}} \sin{\frac{r}{\sqrt{2}}}) + e^{-\frac{r}{\sqrt{2}}} \cos{\frac{r}{\sqrt{2}}} = +e^{-\sqrt{-z}\,r}
    \\
    \sqrt{-z}\,r - \frac{r}{\sqrt{2}} + t_{\theta} (\frac{r}{\sqrt{2}} - e^{-\frac{r}{\sqrt{2}}} \sin{\frac{r}{\sqrt{2}}}) - e^{-\frac{r}{\sqrt{2}}} \cos{\frac{r}{\sqrt{2}}} = -e^{-\sqrt{-z}\,r}  \label{twoeigen2}
\end{align}
For $k \in \rr^3$, the set of all the generalized eigenfunctions of $H_\theta$ is characterised as the following family of the functions:
\begin{equation} \label{generalized eigenfunction}
\psi_{\theta , y_1, y_2}(k , x)=e^{ik \cdot x}+\sum_{m,n=1}^{2}[\Gamma_\theta (|k|^2)]_{mn}^{-1} \, e^{ik \cdot y_{n}} \frac{e^{i\left|k\right|\left|x-y_{m}\right|}}{4 \pi\left|x-y_{m}\right|}
\end{equation}
The scattering length $a_\theta$ of $H_\theta$ as a function of $r$ is
\begin{equation} \label{scattering length}
    \begin{split}
   a_\theta = - \frac{1}{4\pi} \sum_{m,n=1}^{2}[\Gamma_\theta (0)]_{mn}^{-1} = \frac{2 r}{  \frac{r}{\sqrt{2}}  (1-t_\theta) - e^{-r/\sqrt{2}} (t_\theta \sin{\frac{r}{\sqrt{2}}} + \cos{\frac{r}{\sqrt{2}}}) +1}\,. 
    \end{split}
\end{equation}
In particular the scattering length of $H_\theta$ reduces to $\displaystyle  \frac{ \sqrt{2}}{1- t_\theta}$ for $r \rightarrow 0$. \\
\end{lemma}
\begin{proof}
By calculating \eqref{notation} for the two-center case and inserting in the equation \eqref{Matrix-Gamma}, we can obtain matrix \eqref{Gamma-expression} and its inverse explicitly in order to write down the resolvent \eqref{resolvent2} and \eqref{two-becomes-one}. Moreover, its determinant yield to the eigenvalue-resonance equations. The two possible solutions of $\det \Gamma_\theta (z) = 0$ deriving from \eqref{Gamma-expression} are
\[ ( i+t_\theta)s + c_{z}  = \pm (( i+t_\theta)S + C_{z}) \]
These equations lead us to the derivation of \eqref{twoeigen1}, \eqref{twoeigen2}.\\
Generalized eigenfunctions \eqref{generalized eigenfunction} are obtained from the knowledge of the resolvent operator \eqref{resolvent2}
\[ \psi_{\theta , y_1, y_2}(k \omega , x) = \lim_{\substack {x\rightarrow\infty\\ x'/|x'| = \omega}} \lim_{|x'|\rightarrow\infty} 4 \pi |x'| e^{-i(k+i\epsilon)|x'|} \left[H_\theta - (k+i\epsilon)^2 \right]^{-1} (x,x') \]
for $k \geq 0, \omega \in S^2$   and $\det [\Gamma_\theta(k^2)] \neq 0$ (see \cite{RS3}, \cite{Albeverio} for details). \\
 In turn, the knowledge  of the generalized eigenfunctions \eqref{generalized eigenfunction} allows to analize the zero-energy limit of the scattering amplitude and in particular to calculate the scattering length
 \[ f_{\theta, y_1,y_2}(k,\omega,\omega')  =  \lim_{\substack {x\rightarrow\infty\\ x'/|x'| = \omega}} |x| e^{i k \omega'} \left[\Psi_{\theta , y_1, y_2}(k \omega' , x) - e^{i k \omega' \cdot x} \right]  = \frac{1}{4\pi}  \sum^{2}_{m,n =1}  \left[\Gamma_\theta(k^2)\right]^{-1}_{mn} e^{ik(y_m\omega' -y_n \omega)} \]
whose zero energy limit is $- a_\theta$.
\end{proof}
\noindent

In order to clarify the relationship between the extensions $H_\theta$ we 
examined in this section and the definition of finitely many center 
local extensions reported in Section (2.2), we investigate the boundary 
conditions satisfied by the functions in the domain of $H_\theta$ close 
to the interaction centers $y_1, y_2$. Results are stated in the 
following lemma, whose proof is straightforward but cumbersome and will 
be given in the Appendix.
\begin{lemma}
Domain and action of $H_\theta$ can be  characterized as 
follows
\begin{equation}
    \begin{split}
        D(H_\theta)= &\left\{\psi \in L^2(\rr^3)\, | \, \psi = 
\phi^{\lambda}  + q \sum_{j=1}^{2}  \, G^\lambda (\cdot -y_{j})\,, \; \phi^{\lambda} \in H^2(\rr^3)\,,\, q \in \cc, \right. \nonumber\\
    & \left. \,\,\, \phi^{\lambda}(y_j)= \Big( \alpha(r,t_{\theta}) + \frac{\sqrt{\lambda}}{4 \pi} - G^{\lambda}(r)\Big) q \,, \; j=1,2\,, \; \lambda>0 \right\}
    \end{split}
\end{equation}
where
\begin{equation}\label{alpha-description}
     \alpha(r,t_{\theta}) = \frac{ t_{\theta}-1  }{4 \pi\sqrt{2}}   + \frac{e^{-r /\sqrt{2}}}{4 
\pi r} \Big(   t_{\theta}
 \sin (r / \sqrt{2}) + \cos{(r/\sqrt{2})} \Big)\,. 
\end{equation}
Each $\psi \in D(H_\theta)$ has the following behaviour around each 
interaction point \begin{equation} \label{asym-behaviour-equation}
\psi (x \sim y_{i}) = \frac{q}{4 \pi |x-y_{i}|} + \alpha(r,t_{\theta}) \, q + O(1)\,. 
\end{equation}
  The action of $H_\theta$ on its domain is given by
\begin{equation}\label{action}
     \left(H_\theta + \lambda \right) \psi = ( - \Delta + \lambda) 
\phi^{\lambda}\,.
\end{equation}

\begin{remark}
We notice that, taking into account definition \eqref{alpha-description}, the equations $\det \Gamma_{\theta} (-\lambda) =0$ for the negative eigenvalues  $E=-\lambda$, $\lambda>0$, take the form $\tfrac{\sqrt{\lambda}}{4\pi} + \alpha(r,t_{\theta}) \mp G^\lambda (r)=0$, which coincide with the  equations for the eigenvalues in \cite{Albeverio} if one replaces $ \alpha(r,t_{\theta})$ with a constant parameter $\alpha$.
\end{remark}

%
%
%
\end{lemma}

\subsection{Solution of the eigenvalue equation}

Let us consider the equations \eqref{twoeigen1}, \eqref{twoeigen2}  for the negative eigenvalues $E=-\lambda$, $\lambda>0$. 
The equations can be solved in terms of the Lambert function $W$. From \eqref{twoeigen1} we have 

\begin{equation}\label{eqei11}
e^{-\sqrt{\lambda} r } =  \sqrt{\lambda} r  +g_0 \,, \;\;\;\;\;\; g_0=g_0(r/\sqrt{2},t_{\theta}) = (t_{\theta} - 1) \frac{r}{\sqrt{2}} +  e^{-\frac{r}{\sqrt{2}}} \left( t_{\theta}  \sin (\frac{r}{\sqrt{2}}) + \cos (\frac{r}{\sqrt{2}}) \right)   \,.
\end{equation}
With the change of coordinate $s= \sqrt{\lambda} r  + g_0$ the equation becomes $s e^s= e^{g_0}$ and then $ s=W(e^{g_0})$.  Therefore, the solution of \eqref{eqei11} is
\begin{equation}\label{sol-eig-lambert}
\sqrt{\lambda_0} = \frac{1}{r} \Big( W( e^{g_0} ) - g_0 \Big)\,, \;\;\;\;\;\; \text{with}\;\; g_0<1\,.
\end{equation}
Notice that the condition $g_0<1$ guarantees that $\sqrt{\lambda_0}>0$ and it is satisfied for any $r>0$ if we choose $t_{\theta}\leq 1$. We conclude that for any $r>0$ and $t_{\theta}\leq 1$ there is a negative eigenvalue of $H_{\theta}$ given by
\begin{equation}\label{E0}
\epsilon_0 (r, t_{\theta}) = -\lambda_0= -\frac{1}{r^2} \Big( W(e^{g_0(r/\sqrt{2}, t_{\theta}) }) - g_0(r/\sqrt{2}, t_{\theta} )\Big)^{\!2}.
\end{equation}

\begin{remark}\label{remBO}
We recall that the dependence of the negative eigenvalues on the distance $r$ is particularly interesting in view of a possible application to a three-body problem in the Born-Oppenheimer approximation. Indeed, let us  consider a three-particle system made of a light particle of mass $m$ interacting via zero-range interactions with two heavy particles of mass $M$. In the Born-Oppenheimer approximation, i.e., for   $m/M$  small, one first finds the eigenvalues for the light particle when the positions of the heavy particles are considered fixed. Then such eigenvalues, which depend on $r$, play the role of effective potentials for the eigenvalue problem of the heavy particles and the corresponding eigenvalues provide  approximate values of the eigenvalues of the three-particle system. 
\end{remark}

\n
Based on considerations made in the previous remark \ref{remBO}, we proceed to characterize the function $\epsilon_0(\cdot, t_{\theta})$. We first note that  $\epsilon_0(\cdot, t_{\theta})$ is infinitely differentiable for $r>0$. Another relevant property is the asymptotic behavior  for $r \to +\infty$ and for $r \to 0$ in the two cases $t_{\theta}=1$ (unitary limit) and $t_{\theta} <1$. Using well known properties of the Lambert function, we find the following result. 

\vs
\n
Case $a)$: $t_{\theta}=1$

\begin{equation}\label{E01}
\epsilon_0(r) \equiv \epsilon_0(r, 1) = \left\{ \begin{aligned}
          & -\frac{W(1)^2}{r^2} + O \big(e^{-r/\sqrt{2}} \big)              \;\;\;    &\text{for} \quad & r\to +\infty\\
         &- \frac{r^2}{16} + O\big( r^3\big)     &\text{for} \quad & r \to 0\,.
        \end{aligned}\;\right.
\end{equation}

\n
Case $b_1)$: $0<t_{\theta}<1$

\begin{equation}\label{E02}
 \epsilon_0(r, t_{\theta}) = \left\{ \begin{aligned}
          & - \frac{(1-t_{\theta})^2}{2} + O\Big( \frac{e^{- (1-t_{\theta}) r/\sqrt{2} }}{r} \Big)              \;\;\;    &\text{for} \quad & r\to +\infty\\
         &- \frac{(1-t_{\theta})^2}{2} + O(r)   &\text{for} \quad & r\to 0\,.
        \end{aligned}\;\right.
\end{equation}

\n
Case $b_2)$: $t_{\theta} \leq 0$

\begin{equation}\label{E03}
\epsilon_0(r, t_{\theta}) = \left\{ \begin{aligned}
          & - \frac{(1-t_{\theta})^2}{2} + O\Big( \frac{e^{-  r/\sqrt{2} }}{r} \Big)              \;\;\;    &\text{for} \quad & r\to +\infty\\
         &- \frac{(1-t_{\theta})^2}{2} + O(r)   &\text{for} \quad & r \to 0\,.
        \end{aligned}\;\right.
\end{equation}

\n
Notice that cases $b_1)$ and $b_2)$ differ only for the first correction to the limit behavior for $r \to +\infty$. \\
Taking into account \eqref{E01}, we conclude that in the unitary limit $t_{\theta}=1$ ''the potential'' $\epsilon_0(r)$ produces an infinite number of negative eigenvalues accumulating at zero. In the next sections  we shall characterize such eigenvalues showing that they follow the geometrical law  typical of the Efimov effect. \\
For $t_{\theta} <1$, see \eqref{E02}, \eqref{E03}, the potential $\epsilon_0(r, t_{\theta})$ tends to the same constant for $r \to +\infty$ and for $r\to 0$. If one subtracts such constant, the resulting potential produces at most a finite number of negative eigenvalues due to exponential decay at infinity. 
\begin{remark}
    So far, we considered the equation \eqref{twoeigen1} and investigated its solution based on different values of $t_{\theta}$. We can analyse the equation \eqref{twoeigen2} in the same way and obtain 
    \begin{equation}
        \sqrt{\lambda_1} = \frac{1}{r} \left(W(-e^{g_1}) - g_1\right)\,, \;\;\;\;\;\; \text{with}\;\; g_1<-1\,.
    \end{equation}
    where $g_1 = (t_{\theta} - 1) \frac{r}{\sqrt{2}} -  e^{-\frac{r}{\sqrt{2}}} \left( t_{\theta}  \sin (\frac{r}{\sqrt{2}}) + \cos (\frac{r}{\sqrt{2}}) \right) $ and the condition $g_1<-1$ is set to ensure $\sqrt{\lambda_1},l>0$. As a result, the cases $a)$ (i.e. unitary limit) and $b_1)$ become irrelevant since the condition $g_1<-1$ is not satisfied. For the case $b_2)$, where $t_{\theta} \leq 0$  we have
    \begin{equation}\label{E1}
\epsilon_1(r, t_{\theta}) = \left\{ \begin{aligned}
           - &\frac{(1-t_{\theta})^2}{2} + O\Big( \frac{e^{-  r/\sqrt{2} }}{r} \Big)              \;\;\;    &\text{for} \quad & r\to +\infty\\
         & t_{\theta} + O(r)   &\text{for} \quad & r \to 0\,.
        \end{aligned}\;\right.
\end{equation}
\end{remark}

\vs

\vs

\section{Eigenvalue problem for a Hamiltonian with inverse square potential regularized near the origin}

\noindent
As we pointed out in remark \ref{remBO}, our focus lies in the examination of the eigenvalue problem for a one-body Schr\"odinger operator
\begin{equation}\label{Ham}
    H = - \Delta + \epsilon_0(r)\,, \;\;\;\;\;\;\;\; D(H)= H^2(\rr^3)\,.
\end{equation}
As indicated by expression \eqref{E01}, $\epsilon_0(r) \equiv \epsilon_0(r, 1)$ exhibits an inverse square behaviour as $r$ approaches infinity, while it behaves as $O(r^2)$ near the origin.\\
Motivated by this, let us define the following Hamiltonian in $L^2(\rr^3)$
\begin{equation}\label{Hamiltonian}
H_V = - \Delta + V (r)\,, \;\;\;\;\;\;\;\; D(H_V)= H^2(\rr^3)\,, 
\end{equation}
where $r=|x|$ and the potential $V$ is given by
\begin{equation} \label{mixpot}
    V(r) = \left\{ \begin{aligned}
          &V_0 (r)             \;\;\;    &\text{if} \quad & 0\leq r \leq r_0\\
         -&\frac{k}{r^2}  &\text{if} \quad & r>r_0
        \end{aligned}\;,\right.
\end{equation}
with $r_0 >0$, $k>1/4$ and $V_0$ is a non-positive bounded function.
\noindent
By Kato-Rellich theorem the Hamiltonians $H$ is self-adjoint and bounded from below, with $\inf \sigma(H_V) > \min\{\inf V_0\,, - k/r_0^2\}$,  and by Weyl's theorem we have $\sigma_{ess}(H_V)=[0, \infty)$ and therefore, $\sigma_{d}(H_V) \subset (\min\{\inf V_0\,, - k/r_0^2\}, 0)$ (see, e.g. \cite{ReSi}).\\
Let us consider the eigenvalue problem for $H_V$ in the subspace of zero angular momentum. In such a case, using spherical coordinates, the eigenfunctions depend only on the radial coordinate $r$ and the eigenvalue problem reads 
\begin{align}\label{eipr}
&-\frac{d^2 \psi}{dr^2} - \frac{2}{r} \frac{d \psi}{dr} + V (r) \psi = E \psi\,,  \quad \quad E<0, \;\;\;\;\;\; \psi \in  H^2(\rr^3) \,.
\end{align} 

\noindent
It is known that for $k>1/4$ there exists an infinite sequence of negative eigenvalues $E_n$ accumulating at zero (see, e.g. \cite{Sc}, sect. 4.6). In the following we give an independent proof of this fact and we also characterize the asymptotic behavior of $E_n$ for $n \rightarrow \infty$. We start proving the result for $V_{0} = 0$ and subsequently extending it to the case $V_0 = \epsilon_0(r)$.

\vspace{.3cm}

\begin{proposition} \label{pro1}
There exist an infinite sequence of negative eigenvalues $E_n$ of problem \eqref{eipr} for  $V_0=0$,  with $E_n \rightarrow 0$ for $n \rightarrow \infty$. Moreover, 
\begin{equation}\label{asei}
E_n = -  \frac{4}{r_0^2} \, e^{ \, \frac{2}{\beta} \left(   \tan^{-1} (2\beta) + \phi_{\beta} -n \pi   \right)} \big(1 + \zeta_n \big)
\end{equation}
where $\beta = \sqrt{k- 1/4}$,  $\phi_{\beta}= \arg \Gamma (1 + i \beta)$ and 
$\;\zeta_n\rightarrow 0\,, \;\; \text{for}\;\; n \rightarrow \infty$. 
\end{proposition}

\begin{proof}
Let us define  $u(r)= r \psi (r)$ and $E = -\lambda^2$ with $\lambda >0$. The eigenvalue problem \eqref{eipr} for $V_0=0$ amounts to find $u \in L^2(\rr^{+})$ such that:
\begin{equation}\label{reducedeq}
    u^{\prime \prime} - (V(r) + \lambda ^2)u = 0 \quad \quad \quad u(0) =0 \,. 
\end{equation}
Specifically
\begin{align}\label{eqb1}
& u^{\prime \prime} =  \lambda ^2 u & \;\text{for} \;\; &r<r_0  \\
& u^{\prime \prime} + \left( \frac{\beta^2 + 1/4}{r^2} - \lambda ^2 \right)u = 0 & \;\text{for} \;\; &r >r_0 .\label{eqb2}
\end{align}
Equation \eqref{eqb2} is a modified Bessel equation of imaginary order $\nu=i\beta$ and then the general solution can be written as a linear combination of the Bessel functions $L_{i\beta}(\lambda r)$ and $K_{i \beta} (\lambda r)$ (see, e.g., \cite{Du}). Then,
\begin{equation}
u(r)= \left\{ 
\begin{aligned}
&A\, \sinh (\lambda r) + B\, \cosh (\lambda r)\,, \;\;\;\;\;\;\;\; &r<r_0\,\\
&C \, \sqrt{r} L_{i\beta} (\lambda r) + D\, \sqrt{r}K_{i\beta}(\lambda r) \,, \;\;\;\;\;\;\;\; &r>r_0\,
\end{aligned}
\right. 
\end{equation}

where $A, B, C, D$ are arbitrary constants. \\
The above function defines a solution of our eigenvalue problem if and only if  $B=0$ (zero at the origin), $C=0$ (decay at infinity) and continuity of the function and of its derivative  at $r=r_0$. The first two conditions imply that
\begin{equation}
u(r)= \left\{ 
\begin{aligned}
&A\, \sinh (\lambda r) \,, \;\;\;\;\;\;\;\; &r<r_0\,\\
&D\, \sqrt{r} K_{i\beta}(\lambda r) \,, \;\;\;\;\;\;\;\; &r>r_0\,
\end{aligned}
\right. 
\end{equation}
Moreover, denoted $\tau= \lambda r_0$, the continuity conditions at $r=r_0$ read  

\begin{equation}\label{lisy}
\left\{ 
\begin{aligned} 
\sinh (\tau) \, &A  & - &  \; \;\; \sqrt{r_0}K_{i\beta}(\tau) \,  D &=0\\
 \lambda \cosh (\tau) & A & - & \;\left( \frac{1}{2 \sqrt{r_0}} K_{i\beta}(\tau) + \lambda\sqrt{r_0} K_{i\beta}^{'}(\tau) \right)  D &=0
\end{aligned}
\right.
\end{equation}
The linear homogeneous system \eqref{lisy} has non-zero solutions if and only if 

\begin{equation}\label{eqtr1}
K_{i\beta}(\tau)  = f(\tau)\, K_{i\beta}^{'}(\tau) \;\;\;\;\;\;\;\; \text{with}\;\;\;\;\;\; f(\tau):=\frac{2\tau \tanh \tau}{2\tau - \tanh \tau}\,.
\end{equation}

\noindent
Notice that the function $f$  satisfies $\; f(\tau)=2\tau + O(\tau^3)$ for $\tau \rightarrow 0$. 
Moreover  for $\;\tau \rightarrow 0$ we also have (see \cite{Du})
\begin{equation}
K_{i\beta}(\tau) = C_{\beta} \sin \!\left( \beta \log \frac{\tau}{2} - \phi_{\beta} \right) + O(\tau^2)\,, \;\;\;\;\;\; 
K_{i\beta}^{'}(\tau) = C_{\beta} \frac{\beta}{\tau}  \cos \!\left( \beta \log \frac{\tau}{2} - \phi_{\beta} \right) + O(\tau)
\end{equation}
where 
\begin{equation}
C_{\beta}= - \sqrt{\frac{\pi}{\beta \sinh (\beta \pi)}} \,.
\end{equation}
%
Then equation \eqref{eqtr1} can be rewritten as 
\begin{equation}\label{bi}
 \sin \!\left( \beta \log \frac{\tau}{2} - \phi_{\beta} \right) - 2 \beta  \cos \!\left( \beta \log \frac{\tau}{2} - \phi_{\beta} \right)  = F(\tau)
\end{equation}
where $F$ is  continuous and $F(\tau)=O(\tau^2) \;$ for $\; \tau \rightarrow 0$. With the change of variable $\theta= \beta \log \frac{\tau}{2} - \phi_{\beta}$, equation \eqref{bi} reads
\begin{equation}\label{bo}
\sin \theta - 2 \beta \cos \theta = F\big( 2 e^{\frac{\phi_{\beta}}{\beta}}  e^{\frac{\theta}{\beta}} \big)
\end{equation}
The l.h.s. of equation \eqref{bo} is a periodic function oscillating between the values $\;\pm \sqrt{1 + 4 \beta^2}$ while the r.h.s. is a continuous function converging to zero for $\theta \rightarrow -\infty$. Therefore there exists an infinite sequence $\theta_n$ of solutions of equation \eqref{bo}, with $\theta_n <0$ and $\lim_n \theta_n =- \infty$. Hence $\tau_n= 2 \,e^{\frac{\phi_{\beta}}{\beta}}  e^{\frac{\theta_n}{\beta}}$ is an infinite sequence of solutions of equation \eqref{bi}, or equation \eqref{eqtr1}, with $\tau_n >0$ and $\lim_n \tau_n = 0$. \\
Thus we have shown that in the case $V_{0}=0$ there is an infinite sequence of negative eigenvalues $\;E_n= -  \tau_n^2 / r_0^2,\;$ with $\;\lim_n E_n=0$.  The eigenfunction corresponding to the eigenvalue $E_n$ is

\begin{equation}
\psi_n(r)= \left\{ 
\begin{aligned}
&D\, \frac{\sqrt{r_0} \,K_{i\beta}(\tau_n) }{  \sinh \tau_n} \,
 \frac{\sinh (\sqrt{|E_n|} r)}{r} \,, \;\;\;\;\;\;\;\; &r<r_0\,\\
&D\, \frac{K_{i\beta} (\sqrt{|E_n|} r)}{\sqrt{r}} \,, \;\;\;\;\;\;\;\; &r>r_0\,
\end{aligned}
\right. 
\end{equation}
where $D$ is a normalization constant. \\
Let us characterize the asymptotic behavior of the eigenvalues. Let $\theta_n^0$ be the $n$-th negative solution of equation \eqref{bo} for $F=0$ and let $\tau_n^0$ be the corresponding solution of equation \eqref{bi} for $F=0$, i.e.,  
\begin{equation}
\tau_n^0= 2 \, e^{ \,\frac{1}{\beta} \left(   \tan^{-1} (2\beta) + \phi_{\beta} -n \pi   \right)} \;, \;\;\;\;\;\;\;\; n \in N
\end{equation}
Let us write $\tau_n= \tau_n^0 \,e^{\frac{b_n}{\beta}}$. It is sufficient to show that $\lim_n b_n =0$.  By equation \eqref{bi} one has
\begin{equation}
\sin \left( \tan^{-1}(2 \beta) + b_n \right) - 2 \beta \cos \left( \tan^{-1}(2 \beta) + b_n \right) = (-1)^n F(\tau_n)\,.
\end{equation}
Using the identities $\;\sin \!\left( \tan^{-1}(2 \beta) \right) - 2\beta \cos\! \left( \tan^{-1}(2 \beta) \right) =0$ and $\;\cos \!\left( \tan^{-1}(2 \beta) \right) = (1 + 4 \beta^2 )^{-1/2}$, we find 
\begin{equation}
\sin  b_n  = \frac{(-1)^n}{\sqrt{1 + 4 \beta^2}} \, F(\tau_n)\,.
\end{equation}
For $n \rightarrow \infty$ we have $F(\tau_n) \rightarrow 0$, then $b_n \rightarrow 0$.  This concludes the proof.
\end{proof}

\begin{remark}
We note that \eqref{asei} implies the Efimov geometrical law 
\begin{equation}\label{Efimov-geo}
\lim_{n \to \infty} \frac{E_n}{E_{n+1}} = e^{\frac{2 \pi}{\beta}}\,. 
\end{equation}
\end{remark}

 \vspace{.3cm}
 As is clear from formulae \eqref{lisy} and \eqref{eqtr1}, which formalize the requirement of continuity of the solution and its derivative in $r_0$, proposition \ref{pro1} holds true for the solutions to \eqref{reducedeq} with a  potential $V_0 \neq 0$ as far as formula \eqref{eqtr1} holds with $\; f(\tau)=2\tau + O(\tau^3)$ for $\tau \rightarrow 0$.\\ 
In particular, we are interested in extending the validity of proposition  \ref{pro1} to the case where  $V_0 = \epsilon_0(r)$. To this aim, we make use of global  bounds for the solutions of second order differential equations proved  in a very general setting in \cite{bil}. 
\begin{lemma} \label{bound}
 Let $V_{0}$ be a non-positive $\mathcal{C}^1$ function from $\rr^{+}$ to $\rr$. Assume that there exists  $\tilde{r} < r_{0}$ such that for all $\tilde{r} \leq r \leq r_{0}$ one has $V'_{0}(r) \geq 0$ and $|V_{0}(r)|\geq k/r^{2}_{0}$. Then all solutions of \eqref{reducedeq} for $\lambda r_{0} < \sqrt{k}$   satisfy 
  \begin{equation}\label{bound on u}
     u^2(r_{0}) \leq \frac{(u')^{2}(\tilde{r}) r_0^2 - (V_{0}(\tilde{r}) +\lambda^{2}) u^2(\tilde{r}) r_0^2}{ k - r_0^2\lambda^{2}}
 \end{equation}.
 and 
  \begin{equation}
       (u')^2(r_{0}) \leq (u')^{2}(\tilde{r})  - (V_{0}(\tilde{r}) +\lambda^{2}) u^2(\tilde{r})
  \end{equation}
 \end{lemma}
 \begin{proof}
 
The proof is particularly  simple in our specific case and we will report it below for completeness.\\
Multiplying \eqref{reducedeq}  by $u'$ we get
\begin{equation}\label{uno}
     \frac{d (u')^{2}}{dr} - \big(V(r) + \lambda ^2 \big) \frac{d u^{2}}{dr} = 0 \, .
\end{equation}
  Let $0 \leq r_{1} < r_{2} \leq r_{0}$,  integrating equation \eqref{uno} we get
  \begin{equation}\label{due}
  (u')^{2}(r_{1}) - (u')^{2}(r_{2}) = - \big(V_{0}(r_{2}) + \lambda ^2 \big) u^{2}(r_{2}) + \big( V_{0}(r_{1}) + \lambda ^2 \big)u^{2}(r_{1}) + \int_{r_{1}}^{r_{2}} V'_{0}(r) u^{2}(r) dr \, .
\end{equation}

Choosing $r_1 = \tilde{r}$ and $r_2 = r_0$, we can readily deduce based on our assumptions:
\begin{equation}\label{uno'}
       (u')^2(r_{0}) \leq (u')^{2}(\tilde{r})  - (V_{0}(\tilde{r}) +\lambda^{2}) u^2(\tilde{r})
  \end{equation}
  and
\begin{equation}\label{due'}
    u^{2}(r_{0})   \leq \frac{(u')^{2}(\tilde{r}) - (V_{0}(\tilde{r}) +\lambda^{2}) u^2(\tilde{r})}{-(V_{0}(r_{0}) +\lambda^{2})} \leq  \frac{(u')^{2}(\tilde{r}) - (V_{0}(\tilde{r}) +\lambda^{2}) u^2(\tilde{r})}{k/r^2_0 -\lambda^{2}}\,.  
\end{equation}

\end{proof}

Let us now consider  $V_0(r)=\epsilon_0(r)$ in which case  $V_0(r)$ is infinitely differentiable and is of order $O(r^2)$ close to the origin as indicated in \eqref{E01}. Let us denote with $r_{min}$ the value of $r$ where $V_0(r)$ attains its only minimum $V_{min} <0$ (see the blue line in figure 1 in the next section). Being $V^{\prime}_{0}(r) \geq 0$ for any $r \geq r_{min}$, the bound \eqref{bound on u} will be valid for $r_{min} < r <r_0$.
According to \eqref{uno'}, \eqref{due'}, in order to have an estimate of $u(r_0)$ and $u' (r_0)$ we first need to estimate $u$ and $u^{\prime}$ in $r_{min}$. \\
\begin{lemma}\label{bound2}
The solution to \eqref{reducedeq} in $0\leq r \leq r_{min}$ with  initial conditions $u(0)=0\,,\,u'(0)= A \lambda$ satisfy  
\[ |u(r_{min})| \leq  A \lambda r_{min}     \quad    \quad  |u'(r_{min})| \leq A \lambda\]
\end{lemma}
\begin{proof}
Let us first consider equation \eqref{reducedeq} in the interval $I_0 =  [0  , V^{-1}(\lambda^2)]$ where $(V(r) + \lambda^2) \geq 0$ and $u,u'$ and $u''$ have all the same sign. From the equation and from Lagrange theorem we get that there are $\xi_1$ and $\xi_2$ with $\,0 \leq \xi_1 \leq \xi_2 \leq V^{-1}(\lambda^2)$  such that the following equalities and inequalities hold in $I_0$
\begin{equation}
|u'(r)| = |A| \lambda + \lambda^2 |u(\xi_1)| \,V^{-1}(\lambda^2) \quad \quad  |u(\xi_1)| = |u'(\xi_2)| \,V^{-1}(\lambda^2) \quad \quad |u'(\xi_1)| \leq |u'(\xi_2)| \leq |u'(V^{-1}(\lambda^2)| 
\end{equation}
implying that, apart from terms of order $\lambda ^2$, we have $|u'(r)| = A \lambda  u(r)= 0$ for $0\leq r \leq V^{-1}(\lambda^2)$. We are then left to consider \eqref{reducedeq} in $ I_1 = [V^{-1}(\lambda^2) , r_{min}]$  where $(V(r) +\lambda^2) \leq 0$ assuming the same initial conditions we had in $0$.\\
In $I_1$ we have  $u(r) u''(r) \leq 0$ which implies that $|u'|$ is decreasing until $|u'(r)| \neq 0$. Between $0$ and the first zero of $u'$ there is a $\xi_3$ such that
\begin{equation}
|u'(r)| = |A |\lambda + (V(\xi_3)  + \lambda^2)|u(\xi_3)| \geq  |A| \lambda ( 1 - |V(r_{min})| \,r_{min})
\end{equation} 
A simple calculation (see figure 1 for the plot of  $\epsilon_0(r)$) shows that $|V(r_{min})| \,r_{min} << 1$ so that no zeros of $u'(r)$ occurs in $[0, r_{min}]$. We conclude that $|u'(r_{min})| \leq |A| \lambda$ and $|u(r_{min})| \leq |A| \lambda r_{min}$ apart for terms of order $0(\lambda^2)$.
\end{proof}
 Using bounds stated in lemma \ref{bound} with $\tilde{r} = r_{min}$ and \ref{bound2}  we obtain  that
\[|u(r_{0})| \leq C_1 \frac{\lambda r_0}{ \sqrt{k - r_0^2\lambda^{2}}}  \quad \quad \quad |u'(r_{0})| \leq C_2 \lambda \] 
with $C_1$ and $C_2$ positive real constants.\\
We have finally proved that, as expected, solutions of \eqref{reducedeq} in $0 \leq r \leq r_{0}$ with $V_{0} = \epsilon_0(r)$ have the same behaviour close to origin of the corresponding solutions when $V_{0} = 0$. As opposite to the solutions for $V_{0} = 0$  the ones with $V_{0} = \epsilon_0(r)$ oscillate according to the Sturm oscillation theorem (see e.g. \cite{ReSi}) but the estimates on the absolute values of the solution and of its derivative remain valid.\\

\noindent
We can finally determine the structure of the low energy eigenvalues  for  Hamiltonian $H$ \eqref{Ham}. Let us define the following sequence of potentials:
\begin{equation}\label{auxiliary potential}
    V_{m}(r) = \left\{ \begin{aligned}
          &\epsilon_0(r)   &\text{if} \quad & 0\leq r \leq r_m \equiv \sqrt{2}(\frac{3}{4} \pi + m \pi)\\
         -&\frac{W^2(1) }{r^2}  &\text{if} \quad & r>r_m
        \end{aligned}\;,\right.
\end{equation}
We can prove:

\begin{theorem} \label{thfinale}
The Hamiltonian $H$ \eqref{Ham} has an infinite number of negative eigenvalues accumulating at zero whose behaviour near zero is given by \eqref{asei}. In particular they satisfy the Efimov geometrical law \eqref{Efimov-geo}.
\end{theorem}
\begin{proof}
The multiplication operator $V_m$ is a small perturbation of $V$. In fact
\begin{equation}
   \Delta_m \equiv V(r) -V_{m}(r) =            
          \left\{ \begin{aligned}  
         & 0  \;\;\;    &\text{if} \quad & 0\leq r \leq r_m\\
          & \epsilon_0(r) + \frac{W^2(1) }{r^2}  \quad&\text{if} \quad & r>r_m
        \end{aligned}\;,\right.
\end{equation}
It is easy to check that the small oscillation of $\epsilon_0(r)$ around  $\frac{W^2(1) }{r^2}$ for a sufficiently large $m$ satisfies the bound 
\begin{equation}
  \left|  \epsilon_0(r)  + \frac{W^2(1) }{r^2} \right|  \leq \frac{e^{-m\pi}}{ m^2 }     \quad \quad \quad\text{if} \quad r>r_m
\end{equation}

Let us denote $H_m = -\Delta + V_m$. The operator sequence $H_m$ converge in the norm resolvent sense to $H$. Being all the eigen-spaces relative to the negative eigenvalues of $H_m$ and $H$ non-degenerate,  each eigenvalue $\lambda$ of $H$ is the limit of a sequence of eigenvalues $\lambda_m$ of  $H_m$. In fact  the Neumann series 
\begin{equation}
\frac{1}{H - \xi} - \frac{1}{H_m - \xi}  = \sum_{k=1}^{\infty} \frac{1}{H - \xi} \left(  \Delta_m \frac{1}{H - \xi} \right)^k  \quad \quad \text{when} \quad \xi \in \cc \setminus \rr.
\end{equation}
is norm convergent if $\imag \xi > \sup \Delta_m$.  If $\lambda$ is an eigenvalue of $H$ and we take $\xi = \lambda + i \epsilon$ with $\epsilon > 0$ and $\displaystyle \frac{e^{-m\pi}}{m^2 } < \epsilon$  we have 
\begin{equation}
||\frac{1}{H_m - \xi} - \frac{1}{\epsilon}||  \leq \frac{1}{\epsilon}\,\sum_{k=1}^{\infty}  \left(  \frac{e^{-m\pi}}{m^2 \,\epsilon} \right)^k  =  \frac{1}{\epsilon} \frac{e^{-m\pi}}{m^2 \,\epsilon - e^{-m\pi }}.
\end{equation}
Choosing e.g. $\displaystyle \epsilon = \frac{e^{-(m-1)\pi}}{m^2}$ we have
\begin{equation}
\frac{1}{\text{dist}(\sigma(H_m,), \xi )} = ||\frac{1}{H_m - \xi}||   \geq  \frac{1}{\epsilon}  - \frac{e^{-m\pi}}{m^2 \,\epsilon ^2 - e^{-m\pi }\epsilon} > m^2 e^{(m-2)\pi}\,\, .
\end{equation}
We conclude that it exists at least one $ \lambda_n \in \sigma_{disc} (H_m)$  such that
\[ |\lambda - \lambda_n|  \leq |\lambda - \xi| + |\xi - \lambda_n|  \leq \frac{ 2\,\, e^{-(m-2)\pi}}{m^2} \]
and that $\lambda_n$ is in fact the only one in the disc of radius $\epsilon$ around $\lambda$ for m sufficiently large.

\end{proof}

\section{The Born Oppenheimer approximation}
Following the suggestions given at the end of Section 2, we want to further investigate  the dynamics of a three-body system in the limit in which  it is possible to separate the slow dynamics of two heavy, non-interacting, bosons and the fast dynamics of a light particle interacting with the two bosons via zero-range forces.  It is worth mentioning that a rigorous proof that the model we are presenting here is the small mass ratio limit of the dynamics of a system of three particles interacting via contact interactions in three dimensions is still lacking. A sketch  of a tentative proof can be found in \cite{detlef}, where a large part of the results of this paper were outlined. The treatment followed here closely adheres to that of Fonseca et al. \cite {fons} with the difference that we use genuine zero-range interaction of infinite scattering length. In this way, pathologies connected with infinite attractive potentials at small distances do not appear.
\begin{figure}[h]
\centering
\begin{tikzpicture}
 \filldraw[black] (0,0) circle (0.001) node[right] {o};
  \filldraw[black, fill=gray] (0,-2) circle (0.5) node[below, shift={(0,-0.5)}] {1};
  \filldraw[black, fill=gray] (0,2) circle (0.5) node[above, shift={(0,0.5)}] {2};
  \filldraw[black, fill=gray] (-4,2.7) circle (0.2) node[left, shift={(0, 0.4)}] {3};
  
  \draw[->, >=stealth, line width=0.5pt] (0,-2) -- node[midway, below, sloped] {$x+\frac{R}{2}$} (-4,2.7) ;
  \draw[->, >=stealth, line width=0.5pt] (0, 2)  -- node[midway, above, sloped] {$x-\frac{R}{2}$} (-4,2.7);
  \draw[->, >=stealth, line width=1pt] (0,0) -- node[midway, right] {$-\frac{R}{2}$} (0,-2);
  \draw[->, >=stealth, line width=1pt] (0,0) -- node[midway, right] {$\frac{R}{2}$} (0,2);
  \draw[->, >=stealth, line width=1pt] (0,0) -- node[midway, above, shift={(0.2,-0.05)}, sloped] {$x$} (-4,2.7);

\end{tikzpicture}
\end{figure}


The Hamiltonian in the Jacobi coordinates formally reads (see e.g. \cite{fons})
\begin{equation} \label{BOhamil}
H = -\frac{1}{\mu} \triangle_{{R}}  - \frac{1}{\nu} \triangle_{{x}}  + \delta ({x}+ {R}/2) + \delta( {x} -{R}/2) 
\end{equation}
where 
\begin{equation}
 {x} = {x}_3 - \frac{1}{2} ({x}_1 +{x}_2), \;\;\; {R} = {x}_1 - {x}_2 ,\;\;\;   \nu = \frac{2M }{2M+m }, \;\;\;  \mu = \frac{M}{2m} \,. \nonumber 
\end{equation}
Comparison of the coefficients multiplying the kinetic energies in the  Hamiltonian \eqref{BOhamil} suggests that the fast dynamics of the light particle is generated by a two center point interaction Hamiltonian of the type introduced in the previous sections. \\
The mutual interaction between the two bosons acquired as a consequence of their common interaction with the light one is examined through the Born-Oppenheimer approximation.
In this  approximation the analysis of the eigenvalue problem for  the three body system is performed assuming eigenfunctions of the form   

\begin{equation}
\Psi ({x},{R}) = \psi ({x}; {R}) \Phi ({R}) \nonumber
\end{equation}
where $\psi ({x}; {R})$ is the solution of the time independent Schr\"odinger equation for the light particle depending parametrically on ${R}$
\begin{equation} \label{fast}
 \left[ - \frac{1}{\nu} \Delta_{{x}} + \delta({x} + {R}/2) + \delta ({x} - {R}/2) \right] \psi ({x}; {R}) = \epsilon(\left|R\right|) \psi ({x}; {R}) 
\end{equation}
and
 \begin{equation} \label{slow}
 \left[ - \frac{1}{\mu} \Delta_{{R}}  +  \epsilon(\left|R\right|) \right] \Phi ({R}) = E  \Phi ({R})
\end{equation}
where $E$ is the approximate eigenvalue of the three-body system. 

\begin{remark}
If the function $\epsilon(\left|R\right|)$ is computed for  $\alpha=0$  (see equation \eqref{eigeneq}) one finds that $\epsilon(\left|R\right|) \sim - W_0(1)^2/\nu R^2$ both for $R$ near zero and for $R$ tending to infinity. This turns out to imply that the system of the two bosons is unstable for the presence of eigenvalues  going to $-\infty$. In fact all the self-adjoint realizations of a Schr\"odinger operator  with potential  $- \gamma/R^2$, $\gamma > 1/4$, share this pathology \cite{Dere}. It is difficult not to notice the similarities with the instability problems discussed in section 2 relative to  boundary conditions with fixed scattering length. 
\end{remark}

Let us use instead the two-center point interaction Hamiltonian described in the previous section with $t_\theta = 1$. In this case the effective potential $\epsilon_0(R)$ turns out to be regular, bounded everywhere and decaying as $ - W_0(1)^2/ \nu R^2$ at infinity  (see figure below to compare two cases. For simplicity we considered $\nu=1$). When $W_0(1)^2/\nu > 1/4$ theorem \ref{thfinale} applies:  the energy eigenvalues are bounded from below and there are infinitely many low energy eigenstates with eigenvalues accumulating at zero energy. Moreover, they  satisfy the "Efimov scale".   
\begin {figure}[htbp] \nonumber
\centering
\includegraphics[width=15cm]{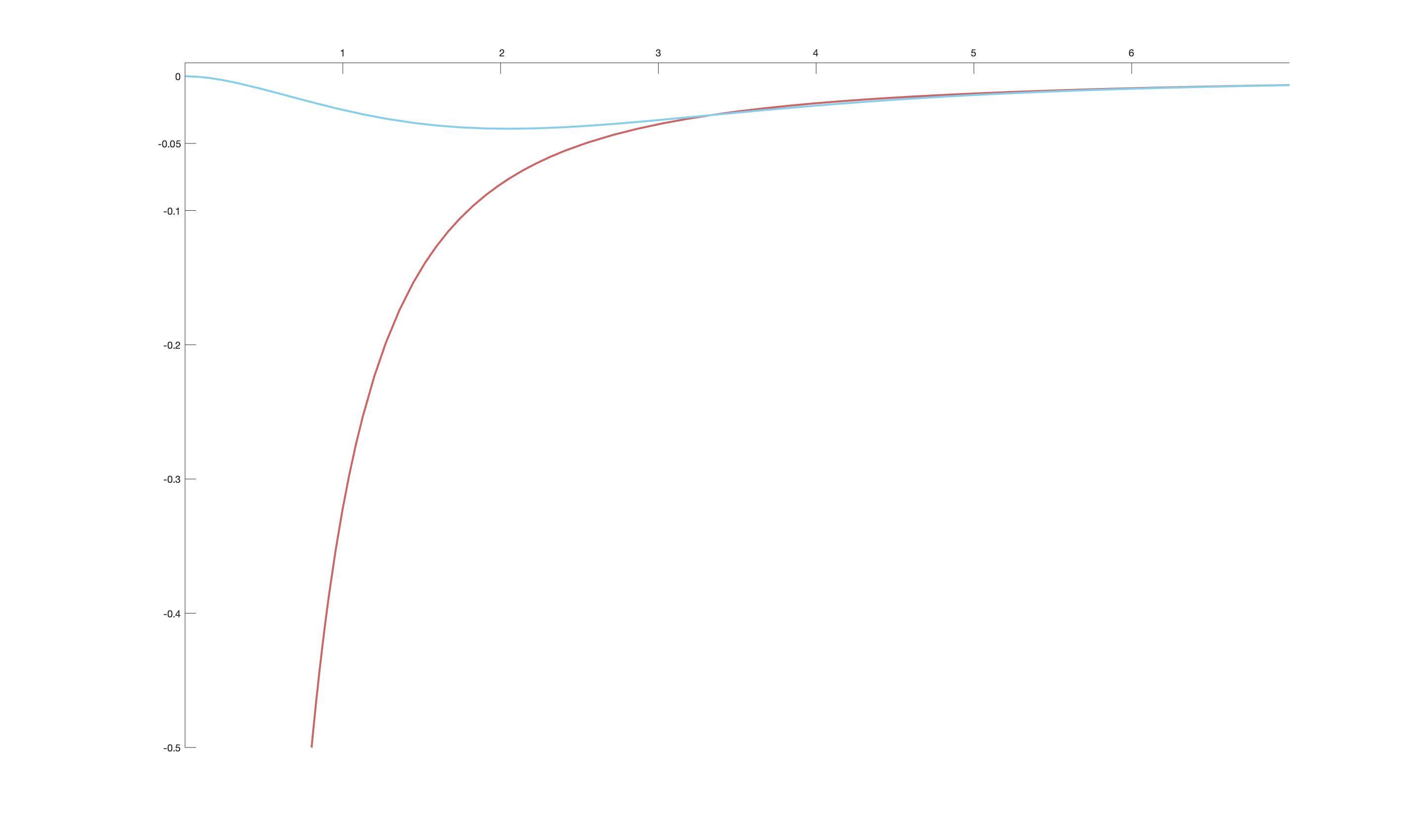}
\caption{\footnotesize {The light blue line represents $\epsilon_{0}(R)$. The lateritious line is the graphic of $\frac{- W^2(1) }{R^2}$. The two curves intersect for $R_{m} = \sqrt{2}(\frac{3\pi}{4} + m \pi),\,\,m=0,1\ldots$ and apart an exponentially decreasing oscillation they coincide  for $R > R_{1} $. }}
\end{figure}

\section{Conclusions}
We investigated a class of many-center point interaction Hamiltonians which had been considered non-relevant or non-physical by physicists and mathematical physicists in the past.  Our aim has been to challenge this view and suggest that these Hamiltonians may be interesting models of quantum systems of particles interacting via short-range interactions. In fact, we proved that the so-called ``non-local" point interactions do not show the problems of non-additivity that ``local" ones do. We also showed that the mechanism making non-local point interaction additive is very similar to the boundary condition renormalization used to avoid ultraviolet catastrophe in systems of three bosons interacting via short-range forces. 

\vspace{1cm}

\begin{appendix}
  \setcounter{equation}{0}
  \renewcommand{\theequation}{\thesection.\arabic{equation}}
   \section{Appendix}
  \label{appA}
\renewcommand{\thesection}{\Alph{section}}

\begin{proof}
As we mentioned at the beginning of section 3, the domain of $H_{\theta}$ in Fourier space is made of functions of the form    
\begin{eqnarray}
D(H_\theta) &=&\left\{f \in L^{2}(\rr^{3}) \left | \right. \hat{f} = \hat{f}_0 +C \left(\frac{1}{(2 \pi)^{3 / 2}} \frac{e^{-i p \cdot y_1}}{p^{2}-i}+ \frac{1}{(2 \pi)^{3 / 2}} \frac{e^{-i p \cdot y_2}}{p^{2}-i}\right) \right.\\
&& \quad \quad \quad \quad \quad \quad \left.+ \,C e^{i \theta}\left(\frac{1}{(2 \pi)^{3 / 2}} \frac{e^{-i p \cdot y_1}}{p^{2}+i} 
+ \frac{1}{(2 \pi)^{3 / 2}} \frac{e^{-i p \cdot y_2}}{p^{2}+i} \right) \right\} \nonumber
\end{eqnarray}
with $C \in \cc , \,\,\, f_{0} \in H^{2}(\rr^{3})$ and $f(y_{1}) = f(y_{2}) =0$.
For $\lambda >0$
$$\frac{1}{p^{2}\pm i}= \underbrace{\frac{\lambda \mp i}{\left(p^{2}\pm i\right)\left(p^{2} + \lambda \right)}}_{\in \mathcal{H}^2}+\frac{1}{p^{2}+\lambda}$$
and the domain can be written as
\begin{eqnarray} \label{domainlanda}
D(H_\theta)=\left\{f \in L^{2}(\rr^{3}) \left | \right. \hat{f} = \hat{f}_0 +C(1 + e^{i \theta}) \left(\frac{1}{(2 \pi)^{3 / 2}} \frac{e^{-i p \cdot y_1}}{p^{2}+ \lambda} +  \frac{1}{(2 \pi)^{3 / 2}} \frac{e^{-i p \cdot y_2}}{p^{2}+ \lambda}\right) \right. +\\ 
 \left.  \frac{C (e^{-i p \cdot y_1} + e^{-i p \cdot y_2}) }{(2 \pi)^{3 / 2}}  
+ \left[\frac{\lambda + i}{\left(p^{2}- i\right)\left(p^{2} + \lambda \right)}\right] + 
\frac{C e^{i \theta}(e^{-i p \cdot y_1} + e^{-i p \cdot y_2}) }{(2 \pi)^{3 / 2}}  \left[ 
\frac{\lambda - i }{\left(p^{2}+ i\right)\left(p^{2} + \lambda \right)}\right]   \right\} \nonumber
\end{eqnarray}

Multiplying  the $H^2$ part of $\hat{f}$ by $e^{i p\cdot y_1}$ and integrating we can compute the value of the regular part of $f$ in $y_1$. Taking into account that $f_0 (y_i) = 0 \,\,\,i=1,2\,\,$   we get
\begin{eqnarray}\label{regular} \nonumber
 f^\lambda_{\text{reg}}(y_1) & = &
C(1 + e^{i \theta}) \left( \frac{\sqrt{\lambda}}{4 \pi} - \frac{1}{1 + e^{i \theta}}\frac{\sqrt{-i}}{4 \pi} - \frac{e^{i \theta}}{ 1 + e^{i \theta}}\frac{\sqrt{i}}{4 \pi}  - \frac{e^{-\sqrt{\lambda} r}}{4 \pi r} +  \frac{1}{ 1 + e^{i \theta}}\frac{e^{-\sqrt{-i} r}}{4 \pi r}  + \frac{e^{i \theta}}{ 1 + e^{i \theta}} \frac{e^{-\sqrt{i} r}}{4 \pi r} \right) \\ \
& = & C(1 + e^{i \theta}) \left( \frac{\sqrt{\lambda}}{4 \pi} - \frac{e^{-\sqrt{\lambda} r}}{4 \pi r}  - (1-t_\theta)\frac{1}{4 \pi \sqrt{2}} + e^{-r/\sqrt{2}} ( \cos(r/\sqrt{2}) + t_\theta \sin(r/\sqrt{2}) \right)
\end{eqnarray}

where the anti-fourier transform was computed for $\real \sqrt{\pm i} > 0$ (see e.g. \cite{Er}).

Equation \eqref{domainlanda} implies that functions in $D(H_\theta)$ have the following behaviour close to the interaction points
\[ f(x) = C(1 + e^{i \theta}) \left(\frac{1}{|x - y_j |} - \frac{\sqrt{\lambda}}{4 \pi} +\frac{e^{-\sqrt{\lambda} r}}{4 \pi r} + O|x - y_j|  \right)  +  f^\lambda_{\text{reg}}(y_j)  \].
Using \eqref{regular} we conclude that
\[f(x) = \frac{q}{4 \pi |x-y_i|} + q \,\alpha(r) + O|x - y_j| \]
where we adopted the traditional notation $q$ for the coefficient of the singular term $q = C(1 + e^{i \theta})$.

Let $f \in D(H_\theta)$. Being $f^\lambda_{\text{reg}} \in \mathcal{H}^2$ we know that  $( - \Delta +\lambda) f^\lambda_{\text{reg}} \in L^2(\rr^3) $. From \eqref{resolvent2} we have
\begin{equation*}
    \left(H_\theta +\lambda \right)^{-1}(-\Delta + \lambda) f^\lambda_{\text{reg}}  = f =f^\lambda_{\text{reg}}  + \sum^{2}_{m,n =1}  \left[\Gamma_\theta(-\lambda)\right]^{-1}_{mn} \left({G_{-\lambda}(\cdot -y_m} \, {,} \, (-\Delta + \lambda) f^\lambda_{\text{reg}} \right) \, G_{-\lambda}(\cdot - y_n)
\end{equation*}
where $\lambda$ must be such that $\det \left[\Gamma_\theta(-\lambda)\right] \neq 0$.
On the other hand,  $\left({G_{-\lambda}(\cdot -y_m} \, {,} \, (-\Delta + \lambda) f^\lambda_{\text{reg}} \right)= f^\lambda_{\text{reg}} (y_m) = f^\lambda_{\text{reg}} (y_n) \equiv c$ . Therefore,
$$\left(H_\theta +\lambda\right)^{-1}(-\Delta - \lambda) f^\lambda_{\text{reg}} = f^\lambda_{\text{reg}}  + q \,\sum_{j=1}^{2} \, G_{\lambda}(\cdot -y_{j}) = f$$
with
\[ q =  c \,\sum^{2}_{n =1}  \left[\Gamma_\theta(-\lambda)\right]^{-1}_{mn} = c \, \sum^{2}_{m =1}  \left[\Gamma_\theta(-\lambda)\right]^{-1}_{mn} \]
or
\begin{equation}\label{q's-expression}
 \begin{pmatrix}
      c \\
     c
 \end{pmatrix}
 =
 \begin{pmatrix}
 [\Gamma_\theta (-\lambda)]_{11} & ( [\Gamma_\theta (-\lambda)]_{12}\\
[\Gamma_\theta (-\lambda)]_{21} &  [\Gamma_\theta (-\lambda)]_{22}     
\end{pmatrix}
      \times
     \begin{pmatrix}
        q\\
         q\\ 
    \end{pmatrix}
    \end{equation}
which is an alternative way to compute the values of the regular part of $f$ using the explicit form \eqref{Gamma-expression} of matrix $\Gamma_\theta (-\lambda)$.
\end{proof}
\end{appendix}

\newpage



\end{document}